\documentclass[superscriptaddress,nofootinbib,notitlepage]{revtex4-1}
\pdfoutput=1
\usepackage{graphicx}
\usepackage{amsthm}
\usepackage{thmtools}
\usepackage{mathtools}
\usepackage{thm-restate}
\usepackage{amsmath}
\usepackage{algorithm} 
\usepackage{algpseudocode} 
\usepackage{amssymb}
\usepackage{bm}
\usepackage{xcolor}
\usepackage{placeins}

\usepackage[normalem]{ulem}
\usepackage[colorlinks]{hyperref}
\hypersetup{
	pdfstartview={FitH},
	pdfnewwindow=true,
	colorlinks=true,
	linkcolor=blue,
	citecolor=blue,
	filecolor=blue,
	urlcolor=blue}
\usepackage{tikz}
\usetikzlibrary{calc,backgrounds}
\usepackage{physics}
\usepackage{cancel}
\usepackage{multirow}
\usepackage{pgffor}
\usepackage{url}
\usepackage{hypcap}
\usepackage{dsfont}
\usepackage{soul}
\usepackage{inconsolata}
\usepackage{hyperref}
\usepackage{footnote}	
\usepackage{qcircuit}
\usepackage{algorithm}
\usepackage{algpseudocode}
\hypersetup{pdfpagemode=UseNone}

\allowdisplaybreaks

\newcommand{\eq}[1]{(\ref{eq:#1})}
\newcommand{\thm}[1]{\hyperref[thm:#1]{Theorem~\ref*{thm:#1}}}
\newcommand{\defn}[1]{\hyperref[defn:#1]{Definition~\ref*{defn:#1}}}
\newcommand{\lem}[1]{\hyperref[lem:#1]{Lemma~\ref*{lem:#1}}}
\newcommand{\prop}[1]{\hyperref[prop:#1]{Proposition~\ref*{prop:#1}}}
\newcommand{\fig}[1]{\hyperref[fig:#1]{Figure~\ref*{fig:#1}}}
\newcommand{\tab}[1]{\hyperref[tab:#1]{Table~\ref*{tab:#1}}}
\renewcommand{\sec}[1]{\hyperref[sec:#1]{Section~\ref*{sec:#1}}}
\newcommand{\app}[1]{\hyperref[app:#1]{Appendix~\ref*{app:#1}}}
\newcommand{\cor}[1]{\hyperref[cor:#1]{Corollary~\ref*{cor:#1}}}
\newcommand{\obs}[1]{\hyperref[obs:#1]{Observation~\ref*{obs:#1}}}
\newcommand{\nn}{\nonumber \\}

\usepackage[capitalise]{cleveref} 

\newtheorem{theorem}{Theorem}

\newtheorem{lemma}[theorem]{Lemma}

\renewcommand{\ket}[1]{|#1\rangle}

\newcommand{\MQ}{\affiliation{%
School of Mathematical and Physical Sciences,
Macquarie University, Sydney, NSW 2109, Australia} }

\begin{document}

\title{Optimum phase estimation with two control qubits}
\author{Peyman Najafi}
\affiliation{%
Donders Institute for Neuroscience, Radboud University, Nijmegen, the Netherlands}
\MQ
\author{Pedro C. S. Costa}\MQ
\author{Dominic W. Berry}
 	\email{Electronic mail: dominic.berry@mq.edu.au}\MQ

\begin{abstract}
Phase estimation is used in many quantum algorithms, particularly in order to estimate energy eigenvalues for quantum systems.
When using a single qubit as the probe (used to control the unitary we wish to estimate the eigenvalue of), it is not possible to measure the phase with a minimum mean-square error.
In standard methods, there would be a logarithmic (in error) number of control qubits needed in order to achieve this minimum error.
Here show how to perform this measurement using only two control qubits, thereby reducing the qubit requirements of the quantum algorithm.
Our method corresponds to preparing the optimal control state one qubit at a time, while it is simultaneously consumed by the measurement procedure.
\end{abstract}
\maketitle

\section{Introduction}

Quantum phase estimation was originally applied in quantum algorithms for the task of period finding, as in Shor's algorithm \cite{Shor}.
Later, quantum phase estimation was applied to the task of estimating eigenvalues for Hamiltonians in quantum chemistry \cite{Alan}.
The appropriate way to perform quantum phase estimation is different between these applications, due to the costing of the operations.
In particular, for estimating eigenvalues, the cost of Hamiltonian simulation is (at least) proportional to the time of evolution, so the phase estimation procedure should attempt to minimise the total evolution time.
At the same time the mean-square error in the estimate should be minimised.

As part of the phase estimation, the inverse quantum Fourier transform is used.
This operation can be decomposed into a `semiclassical' form \cite{griffiths_1996_semiclassical}, where one performs measurements on the control qubits in sequence, with rotations controlled according to the results of previous measurements.
In the form of phase estimation as in Shor's algorithm, the control qubits would be an equal superposition state, which is just a tensor product of $\ket{+}$ states on the individual qubits.
In that scenario, only one control qubit need be used at a time, because it can be prepared in the $\ket{+}$ state, used as a control, then rotated and measured before the next qubit is used.

This procedure with the control qubits in $\ket{+}$ states gives a probability distribution for the error as a sinc function, which has a significant probability for large errors.
That is still suitable for Shor's algorithm, because it is possible to take large powers of the operators with relatively small cost, which suppresses the phase measurement error.
On the other hand, for quantum chemistry where there is a cost of Hamiltonian simulation proportional to time, the large error of the sinc is a problem.
Then it is more appropriate to use qubits in an entangled state \cite{BabbushPRX18}, which was originally derived in an optical context in 1996 \cite{luis1996optimum}.

In 2000 we analysed the problem of how to perform measurements on these states in a Mach-Zehnder interferometer \cite{BerryPRL00}.
The same year, Jon Dowling introduced NOON states in the context of lithography \cite{NOON1}, and then in 2002 showed how NOON states may be used in interferometry for phase measurement \cite{NOON2,NOON3}.
A drawback to using NOON states is that they are highly sensitive to loss.
In 2010 one of us (DWB) visited Jon Dowling's group to work on the problem of how to generate states that are more resistant to loss and effectively perform measurements with them.
This resulted in the publication (separately from Jon) \cite{DinaniPRA14}, followed by our first joint publication \cite{DinaniPRA16}.
We continued collaborating with Jon for many years on phase measurement \cite{HuangPRA17}, as well as state preparation \cite{MotesPRA16}, and Boson-sampling inspired cryptography \cite{Huang2021photonicquantumdata}.

In separate work, we showed how to combine results from multiple NOON states in order to provide highly accurate phase measurements suitable for quantum algorithms \cite{higgins2007entanglement}.
Phase measurement via NOON states is analogous to taking a $\ket{+}$ state and performing a controlled $U^N$ on a target system in quantum computing.
The photons in the arms of the interferometer are analogous to the control qubit in quantum computing, with the phase shift from $U^N$ instead arising from an optical phase shift between the arms of the interferometer.
The NOON state gives very high frequency variation of the probability distribution for the phase, rather than a probability distribution with a single peak.
In 2007 we showed how to combine the results from NOON states with different values of $N$ in order to provide a phase measurement analogous to the procedure giving a sinc distribution in quantum algorithms \cite{higgins2007entanglement}.
(It was experimentally demonstrated with multiple passes through the phase shift rather than NOON states.)

A further advance in \cite{higgins2007entanglement} was to show how to use an adaptive procedure, still with individual $\ket{+}$ states, in order to give the `Heisenberg limited' phase estimate.
That is, rather than the mean-square error scaling as it does for the sinc, it scales as it does for the optimal (entangled) control state.
This procedure still only uses a single control qubit at a time, so is suitable for using in quantum algorithms where the number of qubits available is strongly limited; this is why it was used, for example, in \cite{Kivlichan2020improvedfault}.
On the other hand, although it gives the optimal scaling, the constant factor is not optimal, and improved performance is provided by using the optimal entangled state.

In this paper we show how to achieve the best of both worlds.
That is, we show how to provide the optimal phase estimate (with the correct constant factor), while only increasing the number of control qubits by one.
It is therefore suitable for quantum algorithms with a small number of qubits, while enabling the minimum complexity for a given required accuracy.

In \cref{sec:background} we discuss the optimal state for phase estimation and how its usage can be combined with the semiclassical quantum Fourier transform.
Then in \cref{sec:rec_meth}, we introduce a orthogonal basis of states for subsets of qubits, and prove a recursive form. Finally, in \cref{sec:twoQ} we show how the recursive form can be translated into a sequence of two-qubit unitaries to create the optimal state.

\section{Phase measurement using optimal quantum states}
\label{sec:background}

\subsection{The optimal states}
The optimal states for phase estimation from  \cite{luis1996optimum} are of the form
\begin{equation}\label{eq:optstate}
    \ket{\psi_{\rm opt}} = \sqrt{\frac{2}{N+2}} \sum_{n=0}^{N}\sin\left(\frac{\pi (n+1)}{N +2}\right)\ket{n},
\end{equation}
where $N$ is the total photon number in two modes, and $n$ is the photon number in one of the modes, as for example in a Mach-Zehnder interferometer.
It is also possible to consider the single-mode case where $N$ is a maximum photon number and $\ket{n}$ a Fock state.

In either case a physical phase shift of $\phi$ results in a state of the form
\begin{equation}\label{eq:encpha}
    \ket{\psi_{\rm opt}} = \sqrt{\frac{2}{N+2}} \sum_{n=0}^{N} e^{in\phi} \sin\left(\frac{\pi (n+1)}{N +2}\right)\ket{n}.
\end{equation}
The ideal `canonical' phase measurement is then a positive operator-valued measure (POVM) with elements \cite{canonical}
\begin{equation}
   \frac{N+1}{2\pi} |\check\phi\rangle\langle\check\phi| \, d\phi,
\end{equation}
where
\begin{equation}
    |\check\phi\rangle = \frac 1{\sqrt{N+1}} \sum_{n=0}^{N} e^{in\check\phi} \ket{n}.
\end{equation}
Here we are using $\check\phi$ for the result of the measurement, as distinct from the actual phase $\phi$.
Such a canonical measurement typically cannot be implemented using standard linear optical elements, though it can be approximated with adaptive measurements \cite{BerryPRL00}.

It is easily seen that the error distribution after the measurement is then
\begin{equation}
  \frac{1}{\pi(N+2)}  \left(\frac{\cos((\check\phi-\phi)(1 + N /2)) \sin(\pi/(2 + N))}{
   \cos(\pi/(2 + N)) - \cos(\check\phi-\phi)}\right)^2 .
\end{equation}
In contrast, if one were to use the state with an equal distribution over basis states, then the error probability distribution would be close to a sinc
\begin{equation}
    \frac{1}{2\pi(N+1)} \frac{\sin^2((N+1)(\check\phi-\phi)/2)}{\sin^2((\check\phi-\phi)/2)} .
\end{equation}
The error distributions for these two states are shown in \fig{optplots}.
The central peak for the equal superposition state is a little narrower, but it has large tails in the distribution, whereas the probabilities of large errors for the optimal state are strongly suppressed.

\begin{figure}[tbh]
\centering
\includegraphics[scale=.26]{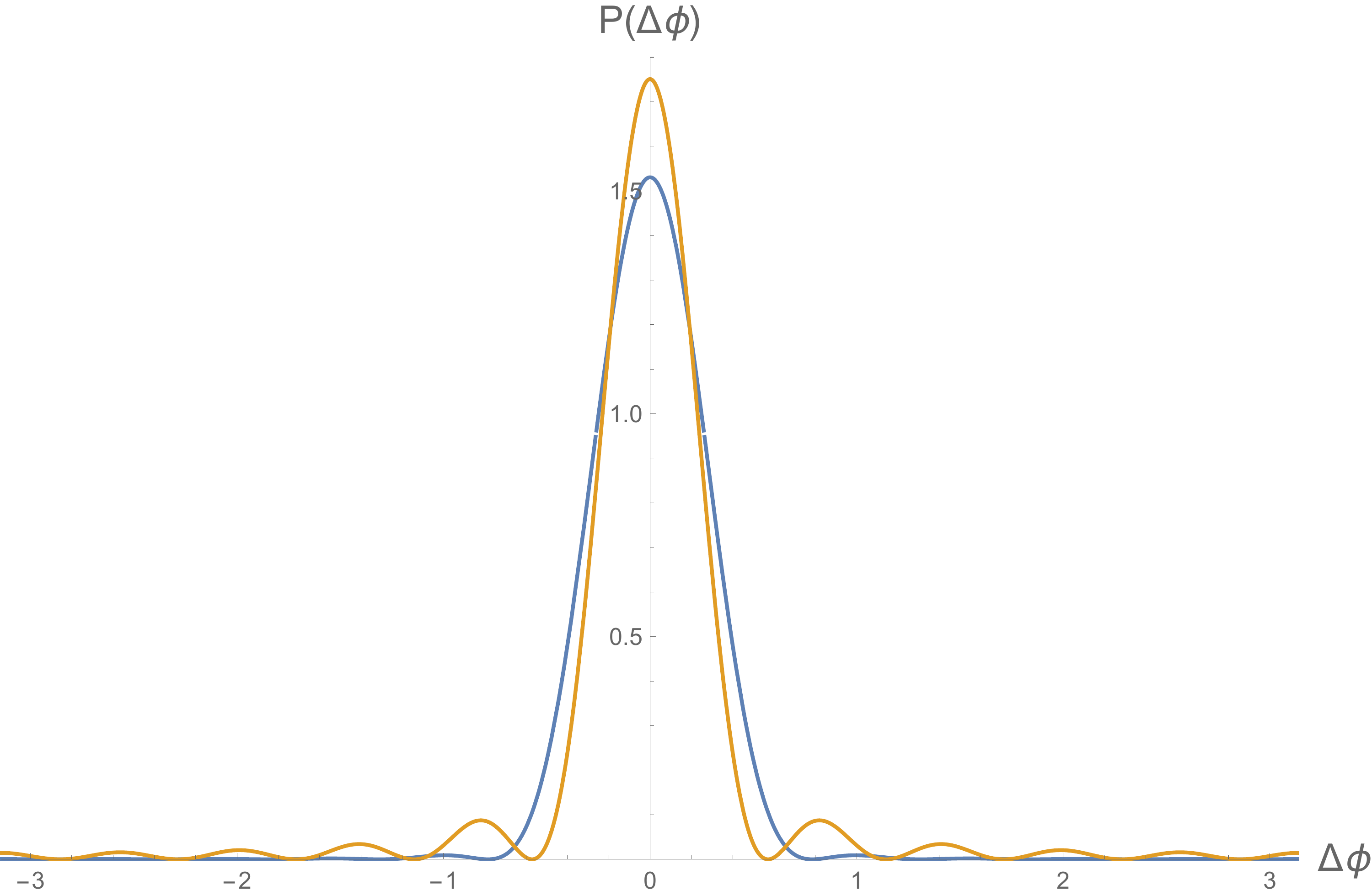} \qquad
\includegraphics[scale=.26]{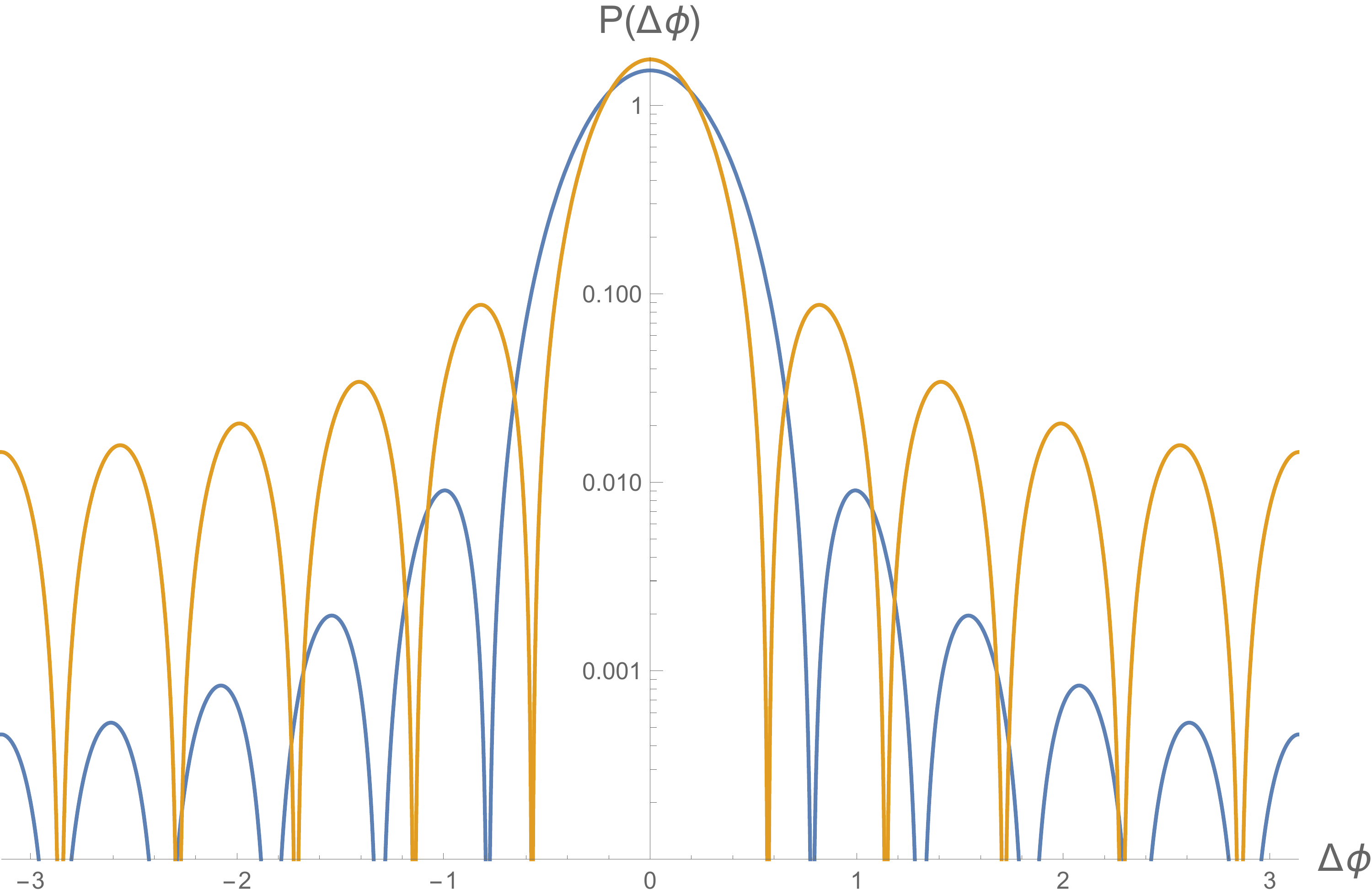}
\caption{\label{fig:optplots} The probability distribution for the error in phase measurements with $N=10$ and the optimal state \eq{optstate} (blue) and the equal superposition state (orange).
The left shows the linear scale and the right shows a log plot.}
\end{figure}

The optimal state \eq{optstate} is optimal for minimising a slightly different measure of error than usual.
The Holevo variance for a distribution can be taken as \cite{Holevo}
\begin{equation}
    |\langle e^{i\check\phi} \rangle |^{-2} -1 .
\end{equation}
This measure has the advantages that it is naturally modulo $2\pi$, as is suitable for phase, and approaches infinity for a flat distribution (with no phase information).
Moreover it approaches the usual variance for suitably narrowly peaked distributions.
To eliminate biased estimate, one can alternatively use the measure
\begin{equation}
    \langle \cos(\check\phi-\phi) \rangle ^{-2} -1 .
\end{equation}
This measure is analogous to the mean-square error.
One could also take the measure, as in \cite{luis1996optimum},
\begin{equation}
    2[1-\langle \cos(\check\phi-\phi) \rangle ] ,
\end{equation}
and the optimisation problem is equivalent.
The optimal state \eq{optstate} gives a minimam Holevo variance of 
\begin{equation}
    \tan^2\left( \frac {\pi}{N+2}\right).
\end{equation}
It is also possible to consider minimisation of the mean-square error, but there is not a simple analytic solution \cite{BerryPRA12}.

\subsection{Phase measurement with the inverse Fourier transform}
In the case of phase measurements in quantum computing, $\phi$ would instead be obtained from a unitary operator $U$ with eigenvalue $e^{i\phi}$.
If the target system is in the corresponding eigenstate of $U$, denoted $\ket{\phi}$, then if state $\ket{n}$ is used to control application of $U^n$, then the $\phi$-dependent state from Eq.~\eq{encpha} is again obtained.
In practice, the integer $n$ is represented in binary in ancilla qubits.
Then the most-significant bit, $n_1$, is used to control $U$, the next most significant bit, $n_2$, is used to control $U^2$, and so forth.
In general,
\begin{equation}
    n = \sum_{j=1}^m n_j 2^{j-1},
\end{equation}
$\ket{n_j}$ is used to control $U^{2^{j-1}}$.
This procedure is depicted in \fig{multPhase}.

\begin{figure}[tbh]
\centerline{
\Qcircuit @R=2em @C=1em {
&\lstick{\ket{n_{1}}}& \qw & \qw & \qw & \cdots && \qw & \ctrl{4} &\qw \\
&\lstick{\ket{n_{2}}}& \qw & \qw & \qw & \cdots && \ctrl{3} &\qw &\qw \\
 \vdots &  & &  &  & & & & &   \\
&\lstick{\ket{n_m}}& \qw &\ctrl{1} & \qw & \cdots && \qw &\qw  &\qw\\
&\lstick{\ket{\phi}} &  \qw & \gate{U^{2^{m-1}}}  & \qw &\cdots& & \gate{U^2} &\gate{U^1}&\qw
}}
\caption{\label{fig:multPhase} The circuit for a controlled power of $U$, where $\ket{n}$ on the control qubits gives $U^n$ on the target register. With the target prepared in an eigenstate $\ket{\phi}$ of $U$, the phase shift $e^{in\phi}$ is obtained.}
\end{figure}
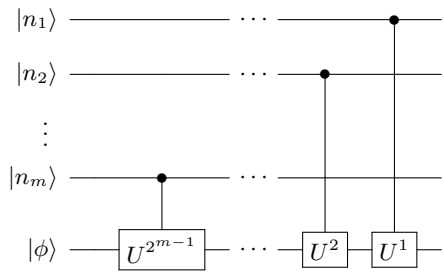

Here we have taken $m$ to be the number of bits.
In practice, it is convenient to take $N-1$ to be a power of 2, so $N=2^m-1$.
In order to estimate the phase, one wishes to perform the canonical measurement on the ancilla qubits.
To explain this, it is convenient to consider the POVM with $N+1$ states $|\check\phi_j\rangle$ with $\check\phi_j=2\pi j/(N+1)$ for $j=0$ to $N$.
Then the states $|\check\phi\rangle$ are mutually orthogonal.
Such a projective measurement can then be obtained if one can perform the unitary operation
\begin{equation}
    \sum_{j=0}^N \ket{j}\langle\check\phi_j| .
\end{equation}
That is, it maps the state $|\check\phi_j\rangle$ to a computational basis state $\ket{j}$, so a measurement in the computational basis gives the result for the phase.
This operation is the inverse of the usual quantum Fourier transform, which would map from $\ket{j}$ to $|\check\phi_j\rangle$.

If one aims to obtain the original POVM, one can randomly (with equal probability) select $\delta\phi\in[0,2\pi/(N+1)]$, and choose the states with $\check\phi_j=2\pi j/(N+1)+\delta\phi$.
Then perform a measurement in the basis $|\check\phi_j\rangle$ with this randomly chosen offset.
The complete measurement, including the random choice of $\delta\phi$, is then equivalent to the POVM with the set of outcomes over a continuous range of $\check\phi$.
This approach can be used in order to give a measurement that is covariant (has an error distribution independent of the system phase $\phi$).
In practice it is not usually needed, so we will not consider it further in this paper.

In order to obtain the estimate for the phase, one should therefore perform the inverse quantum Fourier transform on the control qubits.
The inverse quantum Fourier transform can be performed in a semiclassical way, by performing measurements on successive qubits followed by controlled rotations \cite{griffiths_1996_semiclassical}.
The usual terminology is the `semiclassical Fourier transform', though this is the inverse transform.
An example with three qubits is given in \fig{SCFT}.
The bottom (least significant qubit) is measured first.
The result is used to control a phase rotation on the middle qubit.
Then the middle qubit is measured, and the results of both measurements are used to control phase rotations on the top qubit.
The net result is the same as performing the inverse quantum Fourier transform and measuring in the computational basis.

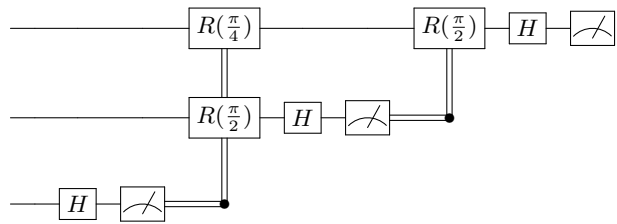
\begin{figure}[tbh]
\centerline{
\Qcircuit @R=2em @C=1em {
& \qw &\qw&\qw & \gate{R(\frac{\pi}{4})} & \qw& \qw &\gate{R(\frac{\pi}{2})} &\gate{H} &\meter\\
 & \qw & \qw& \qw &  \gate{R(\frac{\pi}{2})} \cwx[-1] &\gate{H} &\meter & \control \cw \cwx[-1] \\
 & \qw &\gate{H} & \meter &\control \cw \cwx[-1]  
}}
\caption{\label{fig:SCFT} An example of the semiclassical Fourier transform on three qubits, where the bottom qubit corresponds to the least significant bit.}
\end{figure}

A further advantage of this procedure is that the fact that the controlled $U$ operations are also performed in sequence means that the sequences can be matched.
That is, we have the combined procedure as shown in \fig{USCFT}.
In the case where control registers are prepared in an equal superposition state, then they are unentangled.
This means that preparation of each successive qubit can be delayed until it is needed, as shown in \fig{delayed}.

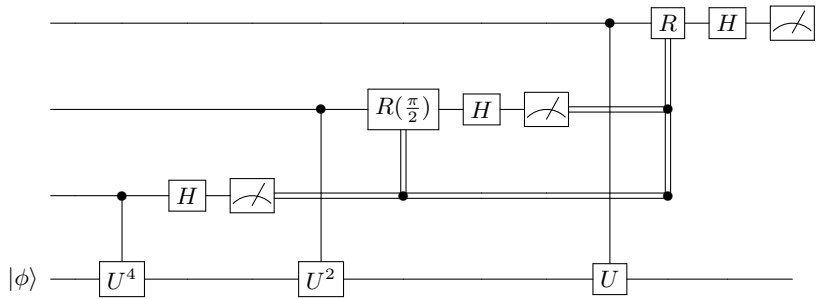
\begin{figure}[tbh]
\centerline{
\Qcircuit @R=2em @C=1em {
& \qw & \qw &\qw&\qw & \qw & \qw & \qw& \qw & \ctrl{3} & \gate{R} &\gate{H} &\meter\\
 & \qw & \qw & \qw & \qw & \ctrl{2} &  \gate{R(\frac{\pi}{2})} &\gate{H} &\meter & \cw & \control \cw \cwx[-1] \\
 & \qw & \ctrl{1} &\gate{H} & \meter & \cw & \control \cw \cwx[-1] & \cw & \cw & \cw & \control \cw \cwx[-1] &    \\
\lstick{\ket{\phi}} &  \qw  & \gate{U^4} & \qw & \qw   & \gate{U^2} & \qw & \qw & \qw &\gate{U}&\qw & \qw & \qw 
}}
\caption{\label{fig:USCFT} The combined procedure with controlled $U^{2^j}$ operations to give phase kickback, together with the semiclassical Fourier transform. The final controlled phase rotation is just shown as $R$, because the angle of rotation is controlled by the combined results of the first two measurements.}
\end{figure}

\begin{figure}[tbh]
\centerline{
\Qcircuit @R=2em @C=1em {
& & & & & & & \lstick{\ket{0}} & \gate{H} & \ctrl{3} & \gate{R} &\gate{H} &\meter\\
 & & & \lstick{\ket{0}} & \gate{H} & \ctrl{2} &  \gate{R(\frac{\pi}{2})} &\gate{H} &\meter & \cw & \control \cw \cwx[-1] \\
\lstick{\ket{0}} & \gate{H} & \ctrl{1} &\gate{H} & \meter & \cw & \control \cw \cwx[-1] & \cw & \cw & \cw & \control \cw \cwx[-1] &    \\
\lstick{\ket{\phi}} &  \qw  & \gate{U^4} & \qw & \qw   & \gate{U^2} & \qw & \qw & \qw &\gate{U}&\qw & \qw & \qw 
}}
\caption{\label{fig:delayed} The combined procedure with preparation of control ancillas delayed until they are needed.}
\end{figure}
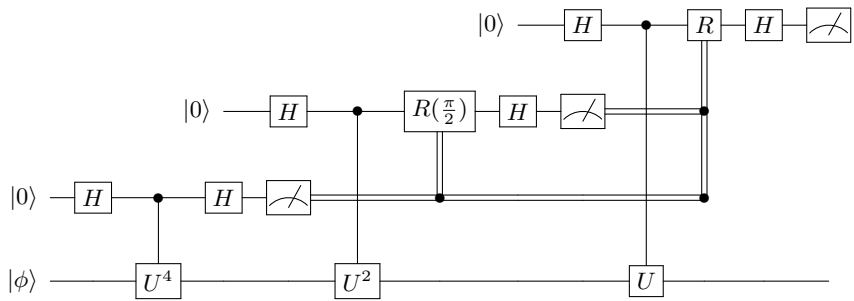

What this means is that only one control qubit need be used at once.
The preparation of the next control qubit can be delayed until after measurement of the previous one, and that qubit can be reset and reused.
That is useful in quantum algorithms with a limited number of qubits available, and is also useful in quantum phase estimation.
In that case, one can replace the control qubits with NOON states with photon numbers that are powers of 2.
Then these NOON states can be measured in sequence to give a canonical measurement of phase, even though a canonical measurement of phase would not be possible on a single two-mode state.
In \cite{higgins2007entanglement} we demonstrated this, using multiple passes through a phase shift rather than NOON states.

The drawback now is that, even though it is possible to perform the canonical measurement, a suboptimal state is being used.
We would like to be able to perform measurements achieving that minimum Holevo phase variance.
In \cite{higgins2007entanglement} we showed that, by using multiple NOON states of each number it is possible to obtain the desired scaling with total photon number, even though there is a different constant factor so the true minimum error is not achieved.

\subsection{Performing phase measurement with two control qubits}

Up until this point this section has been revision of prior work.
What is new here is that we show how to prepare the optimal state for phase measurement in a sequential way, so the number of qubits that need be used at once is minimised.
We will show how the optimal state can be prepared using a sequence of two-qubit operations, as in \fig{twoseq}.

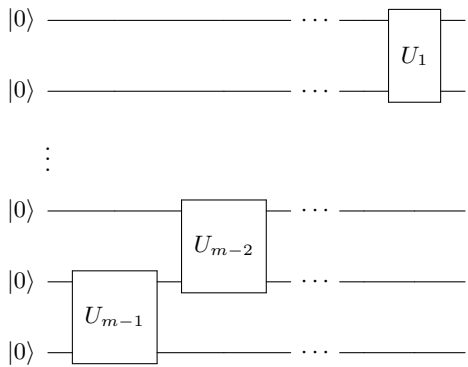
\begin{figure}[tbh]
\centerline{
\Qcircuit @R=2em @C=1em {\lstick{\ket{0}} & \qw & \qw & \qw &\cdots && \qw & \multigate{1}{U_{1}} & \qw \\
\lstick{\ket{0}} &  \qw &\qw & \qw&\cdots && \qw &\ghost{U_{1}} &\qw \\
 \vdots &  & &  &  & & & \\
\lstick{\ket{0}} & \qw & \multigate{1}{U_{m-2}} & \qw &\cdots & & \qw & \qw & \qw   \\
\lstick{\ket{0}} & \multigate{1}{U_{m-1}}  & \ghost{U_{m-2}} & \qw &\cdots & & \qw & \qw & \qw \\
\lstick{\ket{0}} & \ghost{U_{m-1}} & \qw &  \qw &\cdots & & \qw & \qw & \qw 
}}
\caption{\label{fig:twoseq} Preparation of the optimal state with a sequence of two-qubit operations.}
\end{figure}

When the optimal state is prepared in this way, its preparation may be delayed until the qubits are needed, as shown in \fig{twoseqcomb}.
This is illustrated with three control qubits, where introduction of the third qubit can be delayed until the first qubit is measured.
In general, with more control qubits, introduction of each additional qubit can be delayed until after measurement of the qubit two places down, so only two control qubits need be used at once.

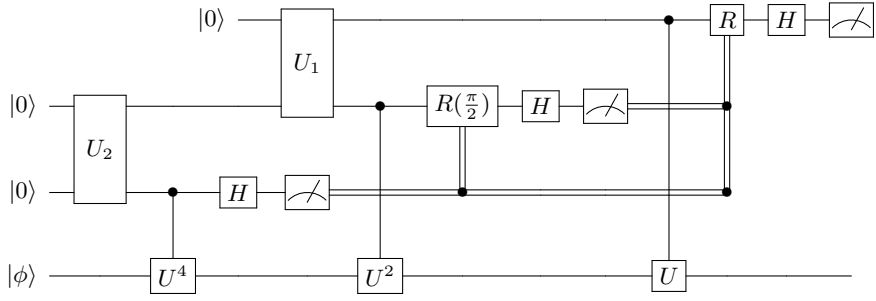
\begin{figure}[tbh]
\centerline{
\Qcircuit @R=2em @C=1em {
& & & \lstick{\ket{0}} & \multigate{1}{U_1} & \qw & \qw & \qw & \qw & \ctrl{3} & \gate{R} &\gate{H} &\meter\\
 \lstick{\ket{0}} & \multigate{1}{U_2} & \qw & \qw & \ghost{U_1} & \ctrl{2} &  \gate{R(\frac{\pi}{2})} &\gate{H} &\meter & \cw & \control \cw \cwx[-1] \\
\lstick{\ket{0}} & \ghost{U_2} & \ctrl{1} &\gate{H} & \meter & \cw & \control \cw \cwx[-1] & \cw & \cw & \cw & \control \cw \cwx[-1] &    \\
\lstick{\ket{\phi}} &  \qw  & \gate{U^4} & \qw & \qw   & \gate{U^2} & \qw & \qw & \qw &\gate{U}&\qw & \qw & \qw 
}}
\caption{\label{fig:twoseqcomb} The combined procedure with preparation of control ancillas in the optimal state delayed until they are needed.}
\end{figure}

The reason why it is possible to prepare the optimal state in this way is that it is a superposition of two unentangled states.
The sine is a combination of a positive and negative complex exponential as
\begin{equation}
    \ket{\psi_{\rm opt}} = \frac 1{2i}\sqrt{\frac{2}{N+2}} \sum_{n=0}^{N} \left[ \exp\left(i\frac{\pi (n+1)}{N +2}\right) - \exp\left(-i\frac{\pi (n+1)}{N +2}\right) \right]\ket{n} .
\end{equation}
When $N=2^m-1$, we can write this as
\begin{align}
\label{eq:opt_state}
 \ket{\psi_{\rm opt}} &= \frac{e^{i\pi/M}}{2i}\sqrt{\frac{2}{M}} \sum_{n_1,\cdots,n_m=0}^{1}e^{i\pi\sum_{j=1}^m n_j2^{j-1}/M}\ket{n_1\cdots n_m} - \frac{e^{-i\pi/M}}{2i}\sqrt{\frac{2}{M}}\sum_{n_1,\cdots,n_m=0}^{1}e^{-i\pi\sum_{j=1}^m n_j2^{j-1}/M}\ket{n_1\cdots n_m}\nn
 &= \frac{e^{i\pi/M}}{2i}\sqrt{\frac{2}{M}} \sum_{n_1,\cdots,n_m=0}^{1} \bigotimes_{j=0}^m e^{i\pi n_j2^{j-1}/N}\ket{n_j} - \frac{e^{-i\pi/N}}{2i}\sqrt{\frac{2}{M}} \sum_{n_1,\cdots,n_m=0}^{1} \bigotimes_{j=0}^m e^{-i\pi n_j2^{j-1}/M}\ket{n_j}\nn
 &= \frac{e^{i\pi/M}}{2i}\sqrt{\frac{2^{m+1}}{M}}  \bigotimes_{j=0}^m \left(\sum_{n_j=0}^{1} \frac{e^{i\pi n_j2^{j-1}/M}}{\sqrt{2}}\ket{n_j}\right) - \frac{e^{-i\pi/M}}{2i}\sqrt{\frac{2^{m+1}}{M}}  \bigotimes_{j=0}^m \left(\sum_{n_j=0}^{1} \frac{e^{-i\pi n_j2^{j-1}/M}}{\sqrt{2}}\ket{n_j}\right) \nn
 &= \sqrt{\frac{2^{m-1}}{M}} \bigotimes_{j=0}^m \left(\sum_{n_j=0}^{1} \frac{e^{-i\pi (-1)^{n_j}2^{j-2}/M}}{\sqrt{2}}\ket{n_j}\right) + \sqrt{\frac{2^{m-1}}{M}}  \bigotimes_{j=0}^m \left(\sum_{n_j=0}^{1} \frac{e^{i\pi (-1)^{n_j}2^{j-2}/M}}{\sqrt{2}}\ket{n_j}\right),
\end{align}
where $M=N+2=2^m+1$.
In the last line we have used
\begin{align}
    \sum_{n_j=0}^{1} \frac{e^{i\pi n_j2^{j-1}/M}}{\sqrt{2}} &= \frac{1}{\sqrt{2}} \left( \ket{0} + e^{i\pi n_j2^{j-1}/M}\right) \nn
    &= e^{i\pi n_j2^{j-2}/M} \frac{1}{\sqrt{2}} \left( e^{-i\pi n_j2^{j-2}/M} \ket{0} + e^{i\pi n_j2^{j-2}/M}\ket{1}\right)
\end{align}
and then
\begin{align}
e^{i\pi/M} \prod_{j=0}^m e^{i\pi n_j2^{j-2}/M} &= e^{i\pi/M} e^{i\pi n_j \sum_{j=0}^m2^{j-2}/M} \nn
&= e^{i\pi/M} e^{i\pi n_j (2^m-1)/2M} \nn
&=  e^{i\pi n_j (2^m+1)/2M} \nn
&=  e^{i\pi n_j /2} \nn &= i.
\end{align}

In order to write the optimal state in a more compact way we define the following
\begin{align}
\label{eq:states+-}
c &\coloneqq \sqrt{\frac{2^{m-1}}{M}}, \nn
\ket{\phi_j^{+}} &\coloneqq \frac{e^{i\pi2^{j-2}/M}\ket{0} + e^{-i\pi2^{j-2}/M}\ket{1}}{\sqrt{2}}, \nn
\ket{\phi_j^{-}} &\coloneqq \frac{e^{-i\pi2^{j-2}/M}\ket{0} + e^{i\pi2^{j-2}/M}\ket{1}}{\sqrt{2}} .
\end{align}
Then the optimal state in \cref{eq:opt_state} can be written as
\begin{equation}
\label{eq:opt_2}
\ket{\psi_{\rm opt}} = c\bigotimes_{j=1}^m \ket{\phi_j^{+}} + c\bigotimes_{j=1}^m \ket{\phi_j^{-}} .
\end{equation}
That is, it is an equally weighted superposition of two states, which are each unentangled between all qubits.
What this means is that any bipartite split of the state will have Schmidt number 2, so the entanglement across the bipartite split can be represented on a single qubit on one side.
We use that principle in the state preparation.
At any stage, after performing the two-qubit operation between qubit $j$ and $j+1$, there will be the correct bipartite entanglement in the split between qubits up to $j$ and qubits from $j+1$ to $m$.
However, at that stage qubits from $1$ to $j-1$ have not been initialised yet, so the entanglement across the bipartite split (for qubits $1$ to $j$) is represented just on qubit $j$.

\section{Recursive construction of the optimum state}\label{sec:rec_meth}

In this section, we show how to create the optimal state \cref{eq:opt_2} recursively.
We introduce the partial tensor product states
\begin{align}
\label{eq:pmphi}
\ket{\phi_{[\ell]}^{+}} &\coloneqq \bigotimes_{j=1}^{\ell}\ket{\phi_j^{+}}, \qquad
\ket{\phi_{[\ell]}^{-}} \coloneqq \bigotimes_{j=1}^{\ell}\ket{\phi_j^{-}}.
\end{align}
Because $\ket{\psi_{\rm opt}}$ is a linear combination of $\ket{\phi_{[m]}^{\pm}}$, the state of qubits 1 to $\ell$ an be represented as a linear combination of $\ket{\phi_{[\ell]}^{\pm}}$.
In order to describe the operations needed to prepare the state $\ket{\psi_{\rm opt}}$, we need to describe the state of qubits 1 to $\ell$ in terms of orthogonal states, which we will denote by $\ket{\Phi_{[\ell]}^{\pm}}$.
These orthogonal (but not normalised) states
\begin{align}
\label{eq:mainStates_opt}
\ket{\Phi_{[\ell]}^+} &\coloneqq \ket{\phi_{[\ell]}^{+}} + \ket{\phi_{[\ell]}^{-}} , \qquad
\ket{\Phi_{[\ell]}^-} \coloneqq \ket{\phi_{[\ell]}^{+}} - \ket{\phi_{[\ell]}^{-}} .
\end{align}
It is possible to prove that these states are orthogonal as in the following Lemma.

\begin{lemma}
The states $\ket{\Phi_{[\ell]}^+}$ and $\ket{\Phi_{[\ell]}^-}$ defined in \cref{eq:mainStates_opt}, are orthogonal:
\begin{equation}
   \langle \Phi_{[\ell]}^- \ket{\Phi_{[\ell]}^+} = 0.
\end{equation}
\end{lemma}
\begin{proof}
From \cref{eq:states+-} we have
\begin{align}
\langle \phi_{j}^- \ket{\phi_{j}^+} &= \cos(\pi2^{j-1}/M) ,
\end{align} 
which is real.
In turn, that implies $\langle \phi_{[\ell]}^- \ket{\phi_{[\ell]}^+}$ is real, and equal to $\langle \phi_{[\ell]}^+ \ket{\phi_{[\ell]}^-}$.
Moreover, because $\ket{\Phi_{j}^\pm}$ are normalised, so are $\ket{\Phi_{[\ell]}^\pm}$.
Therefore we obtain
\begin{align}
\langle \Phi_{[\ell]}^- \ket{\Phi_{[\ell]}^+} &= \langle \phi_{[\ell]}^+ \ket{\phi_{[\ell]}^+}
-\langle \phi_{[\ell]}^- \ket{\phi_{[\ell]}^-}
+\langle \phi_{[\ell]}^+ \ket{\phi_{[\ell]}^-}
-\langle \phi_{[\ell]}^- \ket{\phi_{[\ell]}^+} =0
\end{align}
Here, the first two terms cancel because they are both 1 (due to normalisation) and the second two terms cancel because they are equal (because they are real).
\end{proof}

Next we wish to show that there is a simple recurrence relation for the states $\ket{\Phi_{[\ell]}^\pm}$, in their normalised form 
\begin{equation}\label{eq:mainStates_nrm}
    \ket{\widetilde\Phi_{[\ell]}^\pm} \coloneqq \frac{\ket{\Phi_{[\ell]}^\pm}}{\sqrt{\langle \Phi_{[\ell]}^\pm\ket{\Phi_{[\ell]}^\pm}}}.
\end{equation}
We will use the recurrence relation to derive the sequence of two-qubit operations to prepare the initial state.
The result is as follows.

\begin{lemma}
The states $\ket{\widetilde\Phi_{[\ell]}^+}$ and $\ket{\widetilde\Phi_{[\ell]}^-}$ defined in \cref{eq:mainStates_opt} and \cref{eq:mainStates_nrm}, have recurrence relation
\begin{equation}
    \ket{\widetilde\Phi_{[\ell+1]}^\pm} = \mu_\ell^{(0,\pm)}\ket{\widetilde\Phi_{[\ell]}^\pm} \ket{+} + \mu_\ell^{(1,\pm)}\ket{\widetilde\Phi_{[\ell]}^\mp} \ket{-},
\end{equation}
where
\begin{align}
    \mu_\ell^{(0,\pm)} &=  \cos(2^{\ell-1}\pi/M) P^{0,\pm},\\
    \mu_\ell^{(1,\pm)} &= i\sin(2^{\ell-1}\pi/M) P^{1,\pm},\\
    P^{s,\pm} &= \sqrt{\frac{1 \pm (-1)^s \prod_{j=1}^{\ell} \cos(2^{j-1}\pi/M)}{1 \pm \prod_{j=1}^{\ell+1} \cos(2^{j-1}\pi/M)}}.
\end{align}
\end{lemma}

\begin{proof}
To prove this, we start by noting that
\begin{align}
    \cos( 2^{j-2}\pi/M) &= \langle + \ket{\phi_j^{\pm}}, \nn
    i\sin( 2^{j-2}\pi/M) &= \langle - \ket{\phi_j^{+}}, \nn
    -i\sin( 2^{j-2}\pi/M) &= \langle - \ket{\phi_j^{-}} .
\end{align}
Therefore we can see that
\begin{align}
\langle +\ket{\Phi_{[\ell+1]}^+} &= \cos( 2^{\ell-1}\pi/M)\ket{\Phi_{[\ell]}^+}, \nn
\langle - \ket{\Phi_{[\ell+1]}^+} &= i\sin( 2^{\ell-1}\pi/M)\ket{\Phi_{[\ell]}^-} ,
\end{align}
which implies
\begin{equation}
    \ket{\Phi_{[\ell+1]}^+} = \cos( 2^{\ell-1}\pi/M)\ket{\Phi_{[\ell]}^+} \ket{+} + i\sin( 2^{\ell-1}\pi/M)\ket{\Phi_{[\ell]}^-} \ket{-} .
\end{equation}
Similarly, we find
\begin{align}
\langle +\ket{\Phi_{[\ell+1]}^-} &= \cos( 2^{\ell-1}\pi/M)\ket{\Phi_{[\ell]}^-}, \nn
\langle - \ket{\Phi_{[\ell+1]}^-} &= i\sin( 2^{\ell-1}\pi/M)\ket{\Phi_{[\ell]}^+} ,
\end{align}
which implies
\begin{equation}
    \ket{\Phi_{[\ell+1]}^-} = \cos( 2^{\ell-1}\pi/M)\ket{\Phi_{[\ell]}^-} \ket{+} + i\sin( 2^{\ell-1}\pi/M)\ket{\Phi_{[\ell]}^+} \ket{-} .
\end{equation}
This gives us recurrence relations for $\ket{\Phi_{[\ell]}^\pm}$, which can be written
\begin{equation}
    \ket{\Phi_{[\ell+1]}^\pm} = \cos( 2^{\ell-1}\pi/M)\ket{\Phi_{[\ell]}^\pm} \ket{+} + i\sin( 2^{\ell-1}\pi/M)\ket{\Phi_{[\ell]}^\mp} \ket{-} .
\end{equation}

Let us define the normalisation
\begin{equation}
    {\mathcal{N}}_{[\ell]}^\pm = \sqrt{\langle \Phi_{[\ell]}^\pm \ket{\Phi_{[\ell]}^\pm}}.
\end{equation}
In terms of this, the recurrence relation for the normalised states is
\begin{equation}
    \ket{\widetilde\Phi_{[\ell+1]}^\pm} = \cos( 2^{\ell-1}\pi/M)
    \frac{{\mathcal{N}}_{[\ell]}^+}{{\mathcal{N}}_{[\ell+1]}^\pm}
    \ket{\Phi_{[\ell]}^\pm}\ket{+}  + i\sin( 2^{\ell-1}\pi/M)
    \frac{{\mathcal{N}}_{[\ell]}^-}{{\mathcal{N}}_{[\ell+1]}^\pm}
    \ket{\widetilde\Phi_{[\ell]}^\mp} \ket{-} .
\end{equation}
The normalisation can be determined using
\begin{equation}\label{eq:recur}
    \langle \phi_{[\ell]}^- \ket{\phi_{[\ell]}^+} = \prod_{j=1}^\ell \cos(\pi2^{j-1}/M),
\end{equation}
which gives
\begin{align}
    ({\mathcal{N}}_{[\ell]}^\pm)^2 = \langle \Phi_{[\ell]}^\pm \ket{\Phi_{[\ell]}^\pm} = 2 \pm 2\prod_{j=1}^\ell \cos(\pi2^{j-1}/M).
\end{align}
That gives us the ratios of norms
\begin{equation}
    \frac{{\mathcal{N}}_{[\ell]}^+}{{\mathcal{N}}_{[\ell+1]}^\pm} = P^{0,\pm}, \qquad \frac{{\mathcal{N}}_{[\ell]}^-}{{\mathcal{N}}_{[\ell+1]}^\pm} = P^{1,\pm}.
\end{equation}
Hence \cref{eq:recur} is the form of the recurrence relation required.
\end{proof}

\section{Preparing optimum state with two-qubit unitaries}\label{sec:twoQ}

In the previous section, we showed that it is possible to construct an orthonormal basis for the state on qubits 1 to $\ell$ as $\ket{\widetilde\Phi_{[\ell]}^\pm}$, which satisfy the recursive relation
\begin{align}\label{eq:recursed}
    \ket{\widetilde\Phi_{[\ell+1]}^+} &= \mu_\ell^{(0,+)}\ket{\widetilde\Phi_{[\ell]}^+} \ket{+} + \mu_\ell^{(1,+)}\ket{\widetilde\Phi_{[\ell]}^-} \ket{-} , \nn
    \ket{\widetilde\Phi_{[\ell+1]}^-} &= \mu_\ell^{(0,-)}\ket{\widetilde\Phi_{[\ell]}^-} \ket{+} + \mu_\ell^{(1,-)}\ket{\widetilde\Phi_{[\ell]}^+} \ket{-} \, ,
\end{align}
for constants $\mu_\ell^{(s,\pm)}$.
Moreover, it is obvious from the definitions that $\ket{\widetilde\Phi_{[m]}^+}=\ket{\psi_{\rm opt}}$.
Therefore, the optimum state may be constructed in a recursive way from $\ket{\widetilde\Phi_{[m-1]}^\pm}$ on qubits 1 to $m-1$, which can be constructed from states $\ket{\widetilde\Phi_{[m-2]}^\pm}$ on qubits 1 to $m-2$, and so forth.

In order to prepare the optimal state, we can apply a stepwise procedure where the principle is to use a single qubit to flag which of $\ket{\Phi_{[\ell]}^\pm}$ is to be prepared on the remaining qubits 1 to $\ell$.
We start from qubits $m$ and $m-1$ and work back to qubits 1 and 2.
It is convenient to describe the operations as acting on qubits initialised as $\ket{+}$.
Then we initially perform an operation on qubits $m$ and $m-1$ that maps
\begin{equation}
U_{m-1} \ket{+} \ket{+} = \mu_{m-1}^{(0,+)} \ket{+} \ket{+} + \mu_{m-1}^{(1,+)} \ket{-} \ket{-}.
\end{equation}
The principle of this operation is that it corresponds to the recursion
\begin{equation}
    \ket{\psi_{\rm opt}} = \ket{\widetilde\Phi_{[m]}^+} = \mu_{m-1}^{(1)}\ket{\widetilde\Phi_{[m-1]}^+} \ket{+} + \mu_{m-1}^{(2)}\ket{\widetilde\Phi_{[m-1]}^-} \ket{-} \, .
\end{equation}
At this stage we only have qubits, so the states $\ket{\widetilde\Phi_{[m-1]}^\pm}$ on qubits 1 to $m-1$ need to be represented by $\ket{\pm}$ on a single qubit.

It is trivial to see that unitary $U_{m-1}$ exists; it can explicitly be performed by rotating qubit $m$ as
\begin{equation}
\ket{+} \mapsto \mu_{m-1}^{(0,+)} \ket{+} + \mu_{m-1}^{(1,+)} \ket{-}.
\end{equation}
We can alternatively describe the operation as having the matrix form in the $\ket{\pm}$ basis
\begin{equation}
    U_{m-1} \equiv \begin{pmatrix}
    \mu_{m-1}^{(0,+)} & * & * & * \\
    0 & * & * & * \\
    0 & * & * & * \\
    \mu_{m-1}^{(1,+)} & * & * & * 
    \end{pmatrix},
\end{equation}
where $*$ indicates entries where the value is unimportant.

We then perform $U_{m-2}$ on qubits $m-1$ and $m-2$, down to $U_1$ on qubits 1 and 2.
The unitary $U_\ell$ needs to map
\begin{align}
U_\ell \ket{+}\ket{+} &= \mu_\ell^{(0,+)} \ket{+} \ket{+} + \mu_\ell^{(1,-)} \ket{-} \ket{-}\nn
U_\ell \ket{+}\ket{-} &= \nu_\ell^{(0,+)} \ket{-} \ket{+} + \nu_\ell^{(1,-)} \ket{+} \ket{-}.
\end{align}
This corresponds to the recursion given in \cref{eq:recursed}, and the states $\ket{\pm}$ on qubit $\ell$ are being used to represent $\ket{\widetilde\Phi_{[\ell]}^\pm}$ on qubits 1 to $\ell$.
This operation has the matrix entries
\begin{equation}
    U_{\ell} \equiv \begin{pmatrix}
    \mu_{\ell}^{(0,+)} & 0 & * & * \\
    0 & \mu_{\ell}^{(1,-)} & * & * \\
    0 & \mu_{\ell}^{(0,-)} & * & * \\
    \mu_{\ell}^{(1,+)} & 0 & * & * 
    \end{pmatrix}.
\end{equation}

This operation may be achieved in the following way. Define the single-qubit rotations $V_\ell^{(0,+)}$ and $V_\ell^{(1)}$ to act as
\begin{align}
V_\ell^{(0,+)}\ket{+} &= \mu_\ell^{(0,+)} \ket{+} + \mu_\ell^{(1,+)} \ket{-} \nn
V_\ell^{(0,-)}\ket{+} &= \mu_\ell^{(0,-)} \ket{-} + \mu_\ell^{(1,-)} \ket{+} \, .
\end{align}
Then $U_\ell$ corresponds to the controlled operation
\begin{equation}
U_{\ell} = V_\ell^{(0,+)} \otimes \ket{+}\bra{+}  + V_\ell^{(0,-)} \otimes \ket{-}\bra{-} .
\end{equation}
This method could be used for $U_{m-1}$, though the method described above is simpler.

After performing this sequence of unitaries, we then need to map $\ket{\pm}$ to $\ket{\widetilde\Phi_{[1]}^\pm}$ on qubit 1.
This is a simple single-qubit unitary operation, which can be combined with $U_1$ to give the correct final state with a sequence of two-qubit unitary operations.
Thus our recursive expression for the states gives us a sequence of two-qubit unitaries to create the optimal state.

To be more specific about what operation is needed,
\begin{equation}
    \ket{\phi_{[1]}^+} = \frac{e^{i\pi/2M}\ket{0}+e^{-i\pi/2M}\ket{1}}{\sqrt{2}} ,
\end{equation}
so that
\begin{align}
    \ket{\Phi_{[1]}^+} &= \sqrt{2} \left( \cos(\pi/2M)\ket{0}+\cos(\pi/2M)\ket{1} \right),\\
    \ket{\Phi_{[1]}^-} &= i\sqrt{2} \left( \sin(\pi/2M)\ket{0}-\sin(\pi/2M)\ket{1} \right).
\end{align}
That gives the normalised states
\begin{align}
    \ket{\widetilde\Phi_{[1]}^+} &= \ket{+} ,\qquad
    \ket{\widetilde\Phi_{[1]}^-} = i\ket{-} .
\end{align}
Therefore the operation needed is an $i$ phase shift on $\ket{-}$.

\section{Conclusion}
\label{sec:conc}

We have shown how to create the optimal state for phase estimation, in the sense of minimising Holevo variance, using a sequence of two qubit operations.
When combining this sequential process with the semiclassical quantum Fourier transform, we can entangle new qubits after measuring control qubits in such a way that only two control qubits are needed at once.
This means that the qubit that is measured can be reset and used as the new qubit to be entangled, minimising the need for ancilla qubits.

In quantum algorithms where phase estimation is needed with a small number of logical qubits this is ideal.
Previously the method used was either many entangled qubits, increasing the size of the quantum computer needed, or a single control qubit, which significantly increases the error.
In our method the number of control qubits is only increased by 1, while giving the minimal error.
Here the quantity being exactly minimised is the Holevo variance, which is very close to the mean-square error (MSE) for sharply peaked distributions.
If one were interested in minimising MSE, then these states give the same leading order term for MSE as the minimum MSE \cite{BerryPRA12}, so these states are still suitable.

Our method of preparing the state, although it has been derived for the specific case of the optimal state for minimising Holevo variance, could also be applied to other states that are a superposition of two unentangled states.
The crucial feature is that the Schmidt number is 2 for any bipartite split across the qubits.
One could also consider states with larger Schmidt number, and use a larger number of qubits as controls.
That could potentially be used for states that are optimal for minimising other measures of error.
For example, one could consider methods of approximating Kaiser windows or the digital prolate spheroidal sequence, as is suitable for optimising confidence intervals \cite{Kaiser}.

Another interesting question is whether this procedure could be demonstrated with photons.
A scheme with optimal phase states for $N=2$ using two photons was demonstrated in \cite{DaryanooshNC18}.
With the preparation scheme we have outlined, it would potentially be possible to demonstrate these states with larger $N$, though it would require entangling operations that might require nonlinear optical elements.

\begin{acknowledgments}
DWB worked on this project under a sponsored research agreement with Google Quantum AI. DWB is also supported by Australian Research Council Discovery Projects DP190102633, DP210101367, and DP220101602.
\end{acknowledgments}

\section*{Author Declarations}
\subsection*{Conflict of interest}
The authors have no conflicts to disclose.

\section*{Data availability}
Data sharing is not applicable to this
article as no new data were created or
analyzed in this study.

\end{document}